\documentclass{article}
\usepackage{amsmath,graphicx,./style/spconf}
\usepackage{cite}
\usepackage[T1]{fontenc}
\usepackage{amsmath, amsbsy}
\usepackage{amsthm}
\usepackage{subfigure}
\usepackage{booktabs} 
\usepackage{multirow}
\usepackage{microtype}
\usepackage{balance}
\usepackage{colortbl}
\usepackage{xcolor}
\usepackage{bbm}
\usepackage[hidelinks]{hyperref}
\usepackage{comment}
\usepackage{circledsteps}

\usepackage{amssymb}
\usepackage{amsfonts}
\usepackage{mathrsfs}
\usepackage{xspace}
\usepackage{bm}
\usepackage{upgreek}

\newcommand{\safemath}[2]{\newcommand{#1}{\ensuremath{#2}\xspace}}

\safemath{\bma}{\mathbf{a}}
\safemath{\bmb}{\mathbf{b}}
\safemath{\bmc}{\mathbf{c}}
\safemath{\bmd}{\mathbf{d}}
\safemath{\bme}{\mathbf{e}}
\safemath{\bmf}{\mathbf{f}}
\safemath{\bmg}{\mathbf{g}}
\safemath{\bmh}{\mathbf{h}}
\safemath{\bmi}{\mathbf{i}}
\safemath{\bmj}{\mathbf{j}}
\safemath{\bmk}{\mathbf{k}}
\safemath{\bml}{\mathbf{l}}
\safemath{\bmm}{\mathbf{m}}
\safemath{\bmn}{\mathbf{n}}
\safemath{\bmo}{\mathbf{o}}
\safemath{\bmp}{\mathbf{p}}
\safemath{\bmq}{\mathbf{q}}
\safemath{\bmr}{\mathbf{r}}
\safemath{\bms}{\mathbf{s}}
\safemath{\bmt}{\mathbf{t}}
\safemath{\bmu}{\mathbf{u}}
\safemath{\bmv}{\mathbf{v}}
\safemath{\bmw}{\mathbf{w}}
\safemath{\bmx}{\mathbf{x}}
\safemath{\bmy}{\mathbf{y}}
\safemath{\bmz}{\mathbf{z}}
\safemath{\bmzero}{\mathbf{0}}
\safemath{\bmone}{\mathbf{1}}
\safemath{\Bell}{\ensuremath{\boldsymbol\ell}}

\bmdefine{\biad}{a}
\bmdefine{\bibd}{b}
\bmdefine{\bicd}{c}
\bmdefine{\bidd}{d}
\bmdefine{\bied}{e}
\bmdefine{\bifd}{f}
\bmdefine{\bigd}{g}
\bmdefine{\bihd}{h}
\bmdefine{\biid}{i}
\bmdefine{\bijd}{j}
\bmdefine{\bikd}{k}
\bmdefine{\bild}{l}
\bmdefine{\bimd}{m}
\bmdefine{\bind}{n}
\bmdefine{\biod}{o}
\bmdefine{\bipd}{p}
\bmdefine{\biqd}{q}
\bmdefine{\bird}{r}
\bmdefine{\bisd}{s}
\bmdefine{\bitd}{t}
\bmdefine{\biud}{u}
\bmdefine{\bivd}{v}
\bmdefine{\biwd}{w}
\bmdefine{\bixd}{x}
\bmdefine{\biyd}{y}
\bmdefine{\bizd}{z}

\bmdefine{\bixid}{\xi}
\bmdefine{\bilambdad}{\lambda}
\bmdefine{\bimud}{\mu}
\bmdefine{\bithetad}{\theta}
\bmdefine{\biphid}{\phi}
\bmdefine{\bideltad}{\delta}

\safemath{\bmia}{\biad}
\safemath{\bmib}{\bibd}
\safemath{\bmic}{\bicd}
\safemath{\bmid}{\bidd}
\safemath{\bmie}{\bied}
\safemath{\bmif}{\bifd}
\safemath{\bmig}{\bigd}
\safemath{\bmih}{\bihd}
\safemath{\bmii}{\biid}
\safemath{\bmij}{\bijd}
\safemath{\bmik}{\bikd}
\safemath{\bmil}{\bild}
\safemath{\bmim}{\bimd}
\safemath{\bmin}{\bind}
\safemath{\bmio}{\biod}
\safemath{\bmip}{\bipd}
\safemath{\bmiq}{\biqd}
\safemath{\bmir}{\bird}
\safemath{\bmis}{\bisd}
\safemath{\bmit}{\bitd}
\safemath{\bmiu}{\biud}
\safemath{\bmiv}{\bivd}
\safemath{\bmiw}{\biwd}
\safemath{\bmix}{\bixd}
\safemath{\bmiy}{\biyd}
\safemath{\bmiz}{\bizd}

\safemath{\bmxi}{\bixid}
\safemath{\bmlambda}{\bilambdad}
\safemath{\bmmu}{\bimud}
\safemath{\bmtheta}{\bithetad}
\safemath{\bmphi}{\biphid}
\safemath{\bmdelta}{\bideltad}

\safemath{\bA}{\mathbf{A}}
\safemath{\bB}{\mathbf{B}}
\safemath{\bC}{\mathbf{C}}
\safemath{\bD}{\mathbf{D}}
\safemath{\bE}{\mathbf{E}}
\safemath{\bF}{\mathbf{F}}
\safemath{\bG}{\mathbf{G}}
\safemath{\bH}{\mathbf{H}}
\safemath{\bI}{\mathbf{I}}
\safemath{\bJ}{\mathbf{J}}
\safemath{\bK}{\mathbf{K}}
\safemath{\bL}{\mathbf{L}}
\safemath{\bM}{\mathbf{M}}
\safemath{\bN}{\mathbf{N}}
\safemath{\bO}{\mathbf{O}}
\safemath{\bP}{\mathbf{P}}
\safemath{\bQ}{\mathbf{Q}}
\safemath{\bR}{\mathbf{R}}
\safemath{\bS}{\mathbf{S}}
\safemath{\bT}{\mathbf{T}}
\safemath{\bU}{\mathbf{U}}
\safemath{\bV}{\mathbf{V}}
\safemath{\bW}{\mathbf{W}}
\safemath{\bX}{\mathbf{X}}
\safemath{\bY}{\mathbf{Y}}
\safemath{\bZ}{\mathbf{Z}}

\safemath{\bZero}{\mathbf{0}}
\safemath{\bOne}{\mathbf{1}}
\safemath{\bDelta}{\mathbf{\Delta}}
\safemath{\bLambda}{\mathbf{\UpLambda}}
\safemath{\bPhi}{\mathbf{\Upphi}}
\safemath{\bSigma}{\mathbf{\Upsigma}}
\safemath{\bOmega}{\mathbf{\Upomega}}
\safemath{\bTheta}{\mathbf{\Uptheta}}

\bmdefine{\biAd}{A}
\bmdefine{\biBd}{B}
\bmdefine{\biCd}{C}
\bmdefine{\biDd}{D}
\bmdefine{\biEd}{E}
\bmdefine{\biFd}{F}
\bmdefine{\biGd}{G}
\bmdefine{\biHd}{H}
\bmdefine{\biId}{I}
\bmdefine{\biJd}{J}
\bmdefine{\biKd}{K}
\bmdefine{\biLd}{L}
\bmdefine{\biMd}{M}
\bmdefine{\biOd}{N}
\bmdefine{\biPd}{O}
\bmdefine{\biQd}{P}
\bmdefine{\biRd}{R}
\bmdefine{\biSd}{S}
\bmdefine{\biTd}{T}
\bmdefine{\biUd}{U}
\bmdefine{\biVd}{V}
\bmdefine{\biWd}{W}
\bmdefine{\biXd}{X}
\bmdefine{\biYd}{Y}
\bmdefine{\biZd}{Z}

\bmdefine{\biDelta}{\Delta}
\bmdefine{\biLambda}{\Lambda}
\bmdefine{\biPhi}{\Phi}
\bmdefine{\biSigma}{\Sigma}
\bmdefine{\biOmega}{\Omega}
\bmdefine{\biTheta}{\Theta}

\safemath{\bimA}{\biAd}
\safemath{\bimB}{\biBd}
\safemath{\bimC}{\biCd}
\safemath{\bimD}{\biDd}
\safemath{\bimE}{\biEd}
\safemath{\bimF}{\biFd}
\safemath{\bimG}{\biGd}
\safemath{\bimH}{\biHd}
\safemath{\bimI}{\biId}
\safemath{\bimJ}{\biJd}
\safemath{\bimK}{\biKd}
\safemath{\bimL}{\biLd}
\safemath{\bimM}{\biMd}
\safemath{\bimN}{\biNd}
\safemath{\bimO}{\biOd}
\safemath{\bimP}{\biPd}
\safemath{\bimQ}{\biQd}
\safemath{\bimR}{\biRd}
\safemath{\bimS}{\biSd}
\safemath{\bimT}{\biTd}
\safemath{\bimU}{\biUd}
\safemath{\bimV}{\biVd}
\safemath{\bimW}{\biWd}
\safemath{\bimX}{\biXd}
\safemath{\bimY}{\biYd}
\safemath{\bimZ}{\biZd}

\safemath{\bimDelta}{\biDelta}
\safemath{\bimLambda}{\biLambda}
\safemath{\bimPhi}{\biPhi}
\safemath{\bimSigma}{\biSigma}
\safemath{\bimOmega}{\biOmega}
\safemath{\bimTheta}{\biTheta}

\safemath{\setA}{\mathcal{A}}
\safemath{\setB}{\mathcal{B}}
\safemath{\setC}{\mathcal{C}}
\safemath{\setD}{\mathcal{D}}
\safemath{\setE}{\mathcal{E}}
\safemath{\setF}{\mathcal{F}}
\safemath{\setG}{\mathcal{G}}
\safemath{\setH}{\mathcal{H}}
\safemath{\setI}{\mathcal{I}}
\safemath{\setJ}{\mathcal{J}}
\safemath{\setK}{\mathcal{K}}
\safemath{\setL}{\mathcal{L}}
\safemath{\setM}{\mathcal{M}}
\safemath{\setN}{\mathcal{N}}
\safemath{\setO}{\mathcal{O}}
\safemath{\setP}{\mathcal{P}}
\safemath{\setQ}{\mathcal{Q}}
\safemath{\setR}{\mathcal{R}}
\safemath{\setS}{\mathcal{S}}
\safemath{\setT}{\mathcal{T}}
\safemath{\setU}{\mathcal{U}}
\safemath{\setV}{\mathcal{V}}
\safemath{\setW}{\mathcal{W}}
\safemath{\setX}{\mathcal{X}}
\safemath{\setY}{\mathcal{Y}}
\safemath{\setZ}{\mathcal{Z}}
\safemath{\emptySet}{\varnothing}

\safemath{\colA}{\mathscr{A}}
\safemath{\colB}{\mathscr{B}}
\safemath{\colC}{\mathscr{C}}
\safemath{\colD}{\mathscr{D}}
\safemath{\colE}{\mathscr{E}}
\safemath{\colF}{\mathscr{F}}
\safemath{\colG}{\mathscr{G}}
\safemath{\colH}{\mathscr{H}}
\safemath{\colI}{\mathscr{I}}
\safemath{\colJ}{\mathscr{J}}
\safemath{\colK}{\mathscr{K}}
\safemath{\colL}{\mathscr{L}}
\safemath{\colM}{\mathscr{M}}
\safemath{\colN}{\mathscr{N}}
\safemath{\colO}{\mathscr{O}}
\safemath{\colP}{\mathscr{P}}
\safemath{\colQ}{\mathscr{Q}}
\safemath{\colR}{\mathscr{R}}
\safemath{\colS}{\mathscr{S}}
\safemath{\colT}{\mathscr{T}}
\safemath{\colU}{\mathscr{U}}
\safemath{\colV}{\mathscr{V}}
\safemath{\colW}{\mathscr{W}}
\safemath{\colX}{\mathscr{X}}
\safemath{\colY}{\mathscr{Y}}
\safemath{\colZ}{\mathscr{Z}}

\safemath{\opA}{\mathbb{A}}
\safemath{\opB}{\mathbb{B}}
\safemath{\opC}{\mathbb{C}}
\safemath{\opD}{\mathbb{D}}
\safemath{\opE}{\mathbb{E}}
\safemath{\opF}{\mathbb{F}}
\safemath{\opG}{\mathbb{G}}
\safemath{\opH}{\mathbb{H}}
\safemath{\opI}{\mathbb{I}}
\safemath{\opJ}{\mathbb{J}}
\safemath{\opK}{\mathbb{K}}
\safemath{\opL}{\mathbb{L}}
\safemath{\opM}{\mathbb{M}}
\safemath{\opN}{\mathbb{N}}
\safemath{\opO}{\mathbb{O}}
\safemath{\opP}{\mathbb{P}}
\safemath{\opQ}{\mathbb{Q}}
\safemath{\opR}{\mathbb{R}}
\safemath{\opS}{\mathbb{S}}
\safemath{\opT}{\mathbb{T}}
\safemath{\opU}{\mathbb{U}}
\safemath{\opV}{\mathbb{V}}
\safemath{\opW}{\mathbb{W}}
\safemath{\opX}{\mathbb{X}}
\safemath{\opY}{\mathbb{Y}}
\safemath{\opZ}{\mathbb{Z}}
\safemath{\opZero}{\mathbb{O}}
\safemath{\identityop}{\opI}

\safemath{\veca}{\bma}
\safemath{\vecb}{\bmb}
\safemath{\vecc}{\bmc}
\safemath{\vecd}{\bmd}
\safemath{\vece}{\bme}
\safemath{\vecf}{\bmf}
\safemath{\vecg}{\bmg}
\safemath{\vech}{\bmh}
\safemath{\veci}{\bmi}
\safemath{\vecj}{\bmj}
\safemath{\veck}{\bmk}
\safemath{\vecl}{\bml}
\safemath{\vecm}{\bmm}
\safemath{\vecn}{\bmn}
\safemath{\veco}{\bmo}
\safemath{\vecp}{\bmp}
\safemath{\vecq}{\bmq}
\safemath{\vecr}{\bmr}
\safemath{\vecs}{\bms}
\safemath{\vect}{\bmt}
\safemath{\vecu}{\bmu}
\safemath{\vecv}{\bmv}
\safemath{\vecw}{\bmw}
\safemath{\vecx}{\bmx}
\safemath{\vecy}{\bmy}
\safemath{\vecz}{\bmz}

\safemath{\veczero}{\bmzero}
\safemath{\vecone}{\bmone}
\safemath{\vecxi}{\bmxi}
\safemath{\veclambda}{\bmlambda}
\safemath{\vecmu}{\bmmu}
\safemath{\vectheta}{\bmtheta}
\safemath{\vecphi}{\bmphi}
\safemath{\vecdelta}{\bmdelta}

\safemath{\matA}{\bA}
\safemath{\matB}{\bB}
\safemath{\matC}{\bC}
\safemath{\matD}{\bD}
\safemath{\matE}{\bE}
\safemath{\matF}{\bF}
\safemath{\matG}{\bG}
\safemath{\matH}{\bH}
\safemath{\matI}{\bI}
\safemath{\matJ}{\bJ}
\safemath{\matK}{\bK}
\safemath{\matL}{\bL}
\safemath{\matM}{\bM}
\safemath{\matN}{\bN}
\safemath{\matO}{\bO}
\safemath{\matP}{\bP}
\safemath{\matQ}{\bQ}
\safemath{\matR}{\bR}
\safemath{\matS}{\bS}
\safemath{\matT}{\bT}
\safemath{\matU}{\bU}
\safemath{\matV}{\bV}
\safemath{\matW}{\bW}
\safemath{\matX}{\bX}
\safemath{\matY}{\bY}
\safemath{\matZ}{\bZ}
\safemath{\matzero}{\bmzero}

\safemath{\matDelta}{\bDelta}
\safemath{\matLambda}{\bLambda}
\safemath{\matPhi}{\bPhi}
\safemath{\matSigma}{\bSigma}
\safemath{\matOmega}{\bOmega}
\safemath{\matTheta}{\bTheta}

\safemath{\matidentity}{\matI}
\safemath{\matone}{\matO}

\safemath{\rnda}{A}
\safemath{\rndb}{B}
\safemath{\rndc}{C}
\safemath{\rndd}{D}
\safemath{\rnde}{E}
\safemath{\rndf}{F}
\safemath{\rndg}{G}
\safemath{\rndh}{H}
\safemath{\rndi}{I}
\safemath{\rndj}{J}
\safemath{\rndk}{K}
\safemath{\rndl}{L}
\safemath{\rndm}{M}
\safemath{\rndn}{N}
\safemath{\rndo}{O}
\safemath{\rndp}{P}
\safemath{\rndq}{Q}
\safemath{\rndr}{R}
\safemath{\rnds}{S}
\safemath{\rndt}{T}
\safemath{\rndu}{U}
\safemath{\rndv}{V}
\safemath{\rndw}{W}
\safemath{\rndx}{X}
\safemath{\rndy}{Y}
\safemath{\rndz}{Z}

\safemath{\rveca}{\bimA}
\safemath{\rvecb}{\bimB}
\safemath{\rvecc}{\bimC}
\safemath{\rvecd}{\bimD}
\safemath{\rvece}{\bimE}
\safemath{\rvecf}{\bimF}
\safemath{\rvecg}{\bimG}
\safemath{\rvech}{\bimH}
\safemath{\rveci}{\bimI}
\safemath{\rvecj}{\bimJ}
\safemath{\rveck}{\bimK}
\safemath{\rvecl}{\bimL}
\safemath{\rvecm}{\bimM}
\safemath{\rvecn}{\bimN}
\safemath{\rveco}{\bomO}
\safemath{\rvecp}{\bimP}
\safemath{\rvecq}{\bimQ}
\safemath{\rvecr}{\bimR}
\safemath{\rvecs}{\bimS}
\safemath{\rvect}{\bimT}
\safemath{\rvecu}{\bimU}
\safemath{\rvecv}{\bimV}
\safemath{\rvecw}{\bimW}
\safemath{\rvecx}{\bimX}
\safemath{\rvecy}{\bimY}
\safemath{\rvecz}{\bimZ}

\safemath{\rvecxi}{\bmxi}
\safemath{\rveclambda}{\bmlambda}
\safemath{\rvecmu}{\bmmu}
\safemath{\rvectheta}{\bmtheta}
\safemath{\rvecphi}{\bmphi}

\safemath{\rmatA}{\bimA}
\safemath{\rmatB}{\bimB}
\safemath{\rmatC}{\bimC}
\safemath{\rmatD}{\bimD}
\safemath{\rmatE}{\bimE}
\safemath{\rmatF}{\bimF}
\safemath{\rmatG}{\bimG}
\safemath{\rmatH}{\bimH}
\safemath{\rmatI}{\bimI}
\safemath{\rmatJ}{\bimJ}
\safemath{\rmatK}{\bimK}
\safemath{\rmatL}{\bimL}
\safemath{\rmatM}{\bimM}
\safemath{\rmatN}{\bimN}
\safemath{\rmatO}{\bimO}
\safemath{\rmatP}{\bimP}
\safemath{\rmatQ}{\bimQ}
\safemath{\rmatR}{\bimR}
\safemath{\rmatS}{\bimS}
\safemath{\rmatT}{\bimT}
\safemath{\rmatU}{\bimU}
\safemath{\rmatV}{\bimV}
\safemath{\rmatW}{\bimW}
\safemath{\rmatX}{\bimX}
\safemath{\rmatY}{\bimY}
\safemath{\rmatZ}{\bimZ}

\safemath{\rmatDelta}{\bimDelta}
\safemath{\rmatLambda}{\bimLambda}
\safemath{\rmatPhi}{\bimPhi}
\safemath{\rmatSigma}{\bimSigma}
\safemath{\rmatOmega}{\bimOmega}
\safemath{\rmatTheta}{\bimTheta}

\usepackage{amssymb}
\usepackage{amsfonts}
\usepackage{mathrsfs}
\usepackage{xspace}
\usepackage{bm}
\usepackage{fancyref}
\usepackage{textcomp}

\usepackage{multirow}
\usepackage{stmaryrd}

\newenvironment{textbmatrix}{	\setlength{\arraycolsep}{2.5pt}%
								\left[\begin{matrix}}{\end{matrix}\right]%
								\raisebox{0.08ex}{\vphantom{M}}}

\def\be{\begin{equation}}
\def\ee{\end{equation}}
\def\een{\nonumber \end{equation}}
\def\mat{\begin{bmatrix}}
\def\emat{\end{bmatrix}}
\def\btm{\begin{textbmatrix}}
\def\etm{\end{textbmatrix}}

\def\ba#1\ea{\begin{align}#1\end{align}}
\def\bas#1\eas{\begin{align*}#1\end{align*}}
\def\bs#1\es{\begin{split}#1\end{split}}
\def\bg#1\eg{\begin{gather}#1\end{gather}}
\def\bml#1\eml{\begin{multline}#1\end{multline}}
\def\bi#1\ei{\begin{itemize}#1\end{itemize}}

\newcommand{\lefto}{\mathopen{}\left}

\DeclareMathOperator*{\argmin}{arg\;min}		
\DeclareMathOperator*{\argmax}{arg\;max}		
\DeclareMathOperator{\Exop}{\opE}			

\newcommand{\Ex}[2]{\ensuremath{\Exop_{#1}\lefto[#2\right]}} 	

\safemath{\dirac}{\delta}					
\safemath{\krond}{\dirac}					

\safemath{\upto}{\uparrow}
\safemath{\downto}{\downarrow}
\safemath{\iu}{j}							
\safemath{\ev}{\lambda}						
\safemath{\hilseqspace}{l^{2}}				
\newcommand{\banachfunspace}[1]{\setL^{#1}}	
\safemath{\hilfunspace}{\banachfunspace{2}}	
\newcommand{\floor}[1]{\lfloor #1 \rfloor}

\safemath{\SNR}{\textit{SNR}} 				
\safemath{\PAR}{\textit{PAR}} 				
\safemath{\No}{N_0}							
\safemath{\Es}{E_s}							
\safemath{\Eb}{E_b}							
\safemath{\EbNo}{\frac{\Eb}{\No}}
\safemath{\EsNo}{\frac{\Es}{\No}}

\DeclareMathOperator{\CHop}{\ensuremath{\opH}} 
\safemath{\tvir}{\rndh_{\CHop}}				
\safemath{\tvtf}{\rndl_{\CHop}}				
\safemath{\spf}{\rnds_{\CHop}}				
\safemath{\bff}{H_{\CHop}}					

\safemath{\ircf}{r_{h}}						
\safemath{\tftvcf}{r_{s}}					
\safemath{\tfcf}{r_{l}}						
\safemath{\bfcf}{r_{H}}						

\safemath{\tcorr}{c_h}						
\safemath{\scf}{c_{s}}						
\safemath{\tfcorr}{c_{l}}					
\safemath{\fcorr}{c_{H}}						

\safemath{\mi}{I}							
\safemath{\capacity}{C}						

\safemath{\normal}{\mathcal{N}}			
\safemath{\jpg}{\mathcal{CN}}			
\safemath{\mchain}{\leftrightarrow}		

\safemath{\dB}{\,\mathrm{dB}}
\safemath{\dBm}{\,\mathrm{dBm}}
\safemath{\Hz}{\,\mathrm{Hz}}
\safemath{\kHz}{\,\mathrm{kHz}}
\safemath{\MHz}{\,\mathrm{MHz}}
\safemath{\GHz}{\,\mathrm{GHz}}
\safemath{\s}{\,\mathrm{s}}
\safemath{\ms}{\,\mathrm{ms}}
\safemath{\mus}{\,\mathrm{\text{\textmu}s}}
\safemath{\ns}{\,\mathrm{ns}}
\safemath{\ps}{\,\mathrm{ps}}
\safemath{\meter}{\,\mathrm{m}}
\safemath{\mm}{\,\mathrm{mm}}
\safemath{\cm}{\,\mathrm{cm}}
\safemath{\m}{\,\mathrm{m}}
\safemath{\W}{\,\mathrm{W}}
\safemath{\mW}{\, \mathrm{mW}}
\safemath{\J}{\,\mathrm{J}}
\safemath{\K}{\,\mathrm{K}}
\safemath{\bit}{\,\mathrm{bit}}
\safemath{\nat}{\,\mathrm{nat}}

\safemath{\define}{\triangleq}			

\safemath{\equivalent}{\sim}
\safemath{\distas}{\sim}					
\safemath{\sdiff}{\Delta}				

\safemath{\reals}{\mathbb{R}}
\safemath{\positivereals}{\reals_{+}}
\safemath{\integers}{\mathbb{Z}}
\safemath{\posint}{\integers_{+}}
\safemath{\naturals}{\mathbb{N}}
\safemath{\posnaturals}{\naturals_{+}}
\safemath{\complexset}{\mathbb{C}}
\safemath{\rationals}{\mathbb{Q}}

\newcommand*{\fancyrefapplabelprefix}{app}		
\newcommand*{\fancyrefthmlabelprefix}{thm}		
\newcommand*{\fancyreflemlabelprefix}{lem}		
\newcommand*{\fancyrefcorlabelprefix}{cor}		
\newcommand*{\fancyrefdeflabelprefix}{def}		
\newcommand*{\fancyrefproplabelprefix}{prop}		
\newcommand*{\fancyrefexmpllabelprefix}{exmpl}
\newcommand*{\fancyrefalglabelprefix}{alg}		
\newcommand*{\fancyreftbllabelprefix}{tbl}		

\frefformat{vario}{\fancyrefseclabelprefix}{Sec.~#1}
\frefformat{vario}{\fancyrefthmlabelprefix}{Thm.~#1}
\frefformat{vario}{\fancyreftbllabelprefix}{Tbl.~#1}
\frefformat{vario}{\fancyreflemlabelprefix}{Lem.~#1}
\frefformat{vario}{\fancyrefcorlabelprefix}{Corr.~#1}
\frefformat{vario}{\fancyrefdeflabelprefix}{Def.~#1}
\frefformat{vario}{\fancyreffiglabelprefix}{Fig.~#1}
\frefformat{vario}{\fancyrefapplabelprefix}{App.~#1}
\frefformat{vario}{\fancyrefeqlabelprefix}{(#1)}
\frefformat{vario}{\fancyrefproplabelprefix}{Prop.~#1}
\frefformat{vario}{\fancyrefexmpllabelprefix}{Ex.~#1}
\frefformat{vario}{\fancyrefalglabelprefix}{Alg.~#1}

 \newtheorem{prop}{Proposition}

 \newtheorem{ass}{Assumption}
 \newtheorem*{remark}{Remark}

\safemath{\dictab}{[\,\dicta\,\,\dictb\,]}

\safemath{\ysig}{\bmy}
\safemath{\ysighat}{\hat{\ysig}}
\safemath{\ysigdim}{M}
\safemath{\xsig}{\bmx}
\safemath{\xsigdim}{N}
\safemath{\nx}{n_x}
\safemath{\zsig}{\bmz}
\safemath{\zsigdim}{\ysigdim}
\safemath{\rsig}{\bmr}
\safemath{\Adict}{\bA}
\safemath{\Adicttilde}{\widetilde{\Adict}}
\safemath{\Adictdim}{\outputdim\times\xsigdim}
\safemath{\avec}{\bma}
\safemath{\avectilde}{\tilde{\avec}}
\safemath{\Bdict}{\bB}
\safemath{\Bdicttilde}{\widetilde{\Bdict}}
\safemath{\Cdict}{\bC}
\safemath{\cvec}{\bmc}
\safemath{\Ddict}{\bD}
\safemath{\Ddictdim}{\ysigdim\times\xsigdim}
\safemath{\dvec}{\bmd}
\safemath{\Ddicttilde}{\widetilde{\bD}}
\safemath{\Bonb}{\bB}
\safemath{\bvec}{\bmb}
\safemath{\Bonbdim}{\ysigdim\times\ysigdim}
\safemath{\noise}{\bmn}
\safemath{\noisedim}{\ysigim}
\safemath{\err}{\bme}
\safemath{\errdim}{\ysigdim}
\safemath{\errset}{\setE}
\safemath{\nerr}{n_e}
\safemath{\delop}{\bP_\errset}
\safemath{\delopc}{\bP_{{\errset}^c}}

\safemath{\cplxi}{\imath}
\safemath{\cplxj}{\jmath}

\safemath{\dict}{\matD}
\safemath{\inputdim}{N}		
\safemath{\outputdim}{M}		
\safemath{\sparsity}{S}	
\safemath{\inputdimA}{{N_a}}	
\safemath{\inputdimB}{{N_b}}	
\safemath{\elemA}{{n_a}}	
\safemath{\elemB}{{n_b}}	
\safemath{\resA}{\matR_a}	
\safemath{\resB}{\matR_b}	
\safemath{\subD}{\matS} 
\safemath{\subA}{\matS_a} 
\safemath{\subB}{\matS_b} 
\safemath{\dicta}{\matA} 	
\safemath{\dictb}{\matB} 	
\safemath{\hollowS}{H}
\safemath{\hollowA}{H_a}
\safemath{\hollowB}{H_b}
\safemath{\cross}{Z}
\safemath{\coh}{\mu_d}			
\safemath{\coha}{\mu_a}			
\safemath{\cohb}{\mu_b}			
\safemath{\mubs}{\nu}	
\safemath{\cohm}{\mu_m} 
\safemath{\dictset}{\setD}	
\safemath{\dictsetp}{\dictset(\coh,\coha,\cohb)}	
\safemath{\dictsetgen}{\dictset_\text{gen}}
\safemath{\dictsetgenp}{\dictsetgen(\coh)}
\safemath{\dictsetonb}{\dictset_\text{onb}}
\safemath{\dictsetonbp}{\dictsetonb(\coh)}

\safemath{\leftside}{U}
\safemath{\rightsideA}{R_a}
\safemath{\rightsideB}{R_b}

\safemath{\indexS}{\setI_S} 

\safemath{\na}{n_a}			
\safemath{\nb}{n_b}			
\safemath{\coeffa}{p_i}	
\safemath{\coeffb}{q_j}	
\safemath{\seta}{\setP}		
\safemath{\setb}{\setQ}     
\safemath{\setw}{\setW}	
\safemath{\setz}{\setZ}	
\safemath{\cola}{\veca}		
\safemath{\colb}{\vecb}		
\safemath{\cold}{\vecd}		
\safemath{\inputvec}{\vecx} 	
\safemath{\error}{\vece}	
\safemath{\noiseout}{\vecz} 	
\safemath{\inputvecel}{x}
\safemath{\inputveca}{\vecx_a}
\safemath{\inputvecb}{\vecx_b}
\safemath{\outputvec}{\vecy}	
\safemath{\lambdamin}{\lambda_{\mathrm{min}}}

\safemath{\elltwo}{\ell_2}
\safemath{\ellone}{\ell_1}
\safemath{\ellzero}{\ell_0}
\safemath{\ellinf}{\ell_\infty}
\safemath{\ellinftilde}{\ell_{\widetilde\infty}}
\safemath{\licard}{Z(\coh,\coha,\cohb)}
\safemath{\xsol}{\hat{x}}
\safemath{\xbord}{x_b}		
\safemath{\xstat}{x_s}		
\safemath{\xstatLone}{\tilde{x}_s}
\safemath{\order}{\mathcal{O}} 
\safemath{\scales}{\Theta} 
\safemath{\ones}{\mathbf{1}} 
\safemath{\zeroes}{\mathbf{0}} 
\safemath{\thlone}{\kappa(\coh,\cohb)} 
\safemath{\constoneA}{\delta} 
\safemath{\constoneB}{\epsilon} 
\safemath{\nlarge}{L}				   
\safemath{\sumlarge}{S_\nlarge}
\safemath{\maxlarger}{P_\nlarge}	   
\safemath{\Pzero}{\textrm{P0}}	
\safemath{\Pone}{\textrm{P1}}
\safemath{\vecfir}{\vecw}			 
\safemath{\vecsec}{\vecz}
\safemath{\elvecfir}{w}              
\safemath{\elvecsec}{z}				 
\safemath{\nlargefir}{n}
\safemath{\normout}{\gamma}
\safemath{\auxfun}{h}
\safemath{\supp}{\textrm{supp}}

\safemath{\indexa}{\ell}
\safemath{\indexb}{r}
\safemath{\indexc}{i}
\safemath{\indexd}{j}

\safemath{\project}{P}

\newcommand*\bell{\ensuremath{\boldsymbol\ell}}

\definecolor{mycolor7}{rgb}{0.63500,0.07800,0.18400}%
\newcommand*\circled[1]{\Circled[inner color=white, fill color= mycolor7, outer color=mycolor7]{\footnotesize{#1}}}

\newcommand\blfootnote[1]{%
  \begingroup
  \renewcommand\thefootnote{}\footnotetext{#1}%
  \endgroup
}

\newif\ifshowarxiv
\showarxivtrue

\ifshowarxiv
\newenvironment{showarxiv}
{
}
{
}
\excludecomment{showicassp}
\else
\newenvironment{showicassp}
{
}
{
}
\excludecomment{showarxiv}
\fi

\title{\mbox{Bit Error and Block Error Rate Training for ML-Assisted Communication}}
\name{Reinhard Wiesmayr$^{\,\star,1}$, Gian Marti$^{\,\star,1}$, Chris Dick$^2$, Haochuan Song$^3$, and Christoph Studer$^1$\vspace{-.5cm}}
\address{\small $^\star$equal contribution; $^1$ETH Zurich, $^2$NVIDIA, $^3$Southeast University \vspace{-0.05cm} \\ 
\small E-mail: wiesmayr@iis.ee.ethz.ch, marti@iis.ee.ethz.ch, cdick@nvidia.com, hcsong@seu.edu.cn, studer@ethz.ch
\vspace{-.3cm}
}

\begin{document}
\maketitle
\begin{abstract}
Even though machine learning (ML) techniques are being widely used in communications, 
the question of \emph{how} to train communication systems has received surprisingly little attention. 
In this paper, we show that the commonly used binary cross-entropy (BCE) loss is a sensible choice
in uncoded systems, e.g., for training ML-assisted data detectors, but may not be optimal in coded systems. 
We propose new loss functions targeted at minimizing the block error rate and
SNR deweighting, a novel method that trains communication systems for 
optimal performance over a range of signal-to-noise ratios. 
The utility of the proposed loss functions as well as of SNR deweighting is shown through simulations in NVIDIA~Sionna.

\end{abstract}
\begin{showicassp}
\blfootnote{An extended version of this work that includes an appendix with all the proofs,
some information-theoretic comments, and experiments for another communication scenario
is available on arXiv \cite{wiesmayr2022arxiv}.
}
\end{showicassp}
\begin{showarxiv}
\blfootnote{A shorter version of this paper has been submitted to the 2023 IEEE International 
Conference on Acoustics, Speech, and Signal Processing (ICASSP).
}
\end{showarxiv}
\blfootnote{All code and simulation scripts to reproduce the results of this paper are available on GitHub:
{https://github.com/IIP-Group/BLER\_Training}
}
\blfootnote{The authors thank Oscar Casta\~neda for comments and suggestions.}

\vspace{-0.225cm}
\section{Introduction}
\label{sec:intro}
\vspace{-0.225cm}
Machine learning (ML) has revolutionized a large number of fields, including communications.
The availability of software frameworks, such as TensorFlow \cite{tensorflow2015-whitepaper} and, recently, NVIDIA Sionna \cite{Hoydis2022}, has made  implementation and 
training of ML-assisted communication systems convenient. 
Existing results in ML-assisted communication systems range from the atomistic improvement of data detectors 
(e.g., using deep \mbox{unfolding)} \cite{o2017deep, samuel2019learning, khani2020adaptive, BalatsoukasStimming2019}
to model-free learning of end-to-end communication systems \cite{dorner2017deep, aoudia2019model, song2022benchmarking}.
Quite surprisingly, only little attention has been devoted to the question of \emph{how} ML-assisted communication systems 
should be trained. In particular, the choice of the cost function is seldom discussed
(see, e.g., the recent overview papers \cite{ly2021review, albreem2021deep})
and---given the similarity between
communication and classification---one usually resorts to an empirical cross-entropy (CE) loss \cite{gruber2017deep, xu2017improved, nachmani2018deep, Cammerer2020, honkala2021deeprx, song2021soft}. 
The question of training a communication system for good performance over a range of signal-to-noise ratios (SNRs) is another issue that has not been
seriously investigated. 
Systems are usually trained on samples from only one SNR \cite{o2017deep,aoudia2019model}, 
or on samples uniformly drawn from the targeted SNR range \cite{samuel2019learning, nachmani2018deep, honkala2021deeprx}, apparently without
questioning how this may affect performance for different~SNRs. 

In this paper, we investigate how ML-assisted communication systems should be trained.
We first consider the case where the intended goal is to minimize the uncoded bit error rate (BER) and discuss
why the empirical binary cross-entropy (BCE) loss is indeed a sensible choice in uncoded systems,
e.g., for data detectors in isolation.
However, in most practical communication applications, the relevant figure of merit is the (coded) block error rate (BLER), 
as opposed to the BER,
since block errors %
cause undesirable retransmissions \cite[Sec.\,9.2]{Dahlman}, whereas
(coded) bit errors themselves are irrelevant.\footnote{For this reason, physical layer (PHY) quality-of-service is assessed only 
in terms of BLER (not BER) in 3GPP LTE and other standards. 
Reference \cite{Lipovac2014} notes that the relation
between BER and BLER can be inconsistent.} 
We underpin that minimizing the (coded) BER is not equivalent to minimizing the BLER.
This observation calls into question the common practice
of training coded systems with loss functions that penalize individual bit errors (such as the empirical BCE),
and thus optimize for the (irrelevant) coded BER instead of the BLER.
In response, we propose a range of novel loss functions that aim at minimizing the BLER by penalizing bit errors \emph{jointly}. 
We also show that training on samples that are uniformly drawn from a target SNR range will focus primarily on 
the low-SNR region while neglecting high-SNR performance. As a remedy, we propose a new technique called
SNR deweighting. 
We evaluate the impact of the different loss functions as well as of SNR deweighting through simulations in NVIDIA Sionna \cite{Hoydis2022}.
\begin{showicassp}%
All proofs, as well as additional analysis and experiments, are included in the extended arXiv version \cite{wiesmayr2022arxiv}.
\end{showicassp}

\vspace{-0.275cm}
\section{Training for Bit Error Rate}\label{sec:ber_training}
\vspace{-0.275cm}

ML-assisted communication systems are typically trained with a focus on minimizing the (uncoded) BER \cite{samuel2019learning, honkala2021deeprx}, 
under a tacit assumption that the learned system could then be used in combination with a
forward error correction (FEC) scheme to ensure reliable communication.\footnote{
The discussion also applies to systems that already include FEC, 
but we argue in Secs.~1 and~3 that minimizing the coded BER is a category mistake.}
Due to the similarity between detection and classification, the strategy 
typically consists of (approximately) minimizing the empirical BCE\footnote{When
we speak of the BCE between vectors, we mean the sum of binary CEs
between the individual components as defined in \eqref{eq:lbce}, 
and not the categorical CE between the bit-vector and its estimate (as used, e.g., in \cite{dorner2017deep, aoudia2019model, song2022benchmarking}).
}
on a training set $\mathcal{D}=\{(\bmb^{(n)},\bmy^{(n)})\}_{n=1}^{N}$, 
where $\bmb=(b_1,\dots,b_K)$ is the vector of bits of interest 
(even in uncoded systems, one is interested in multiple bits, e.g., when using higher-order constellations, multiple OFDM subcarriers, 
or multi-user transmission),
$\bmy\in\setY$ is the channel output, and $n$ is the sample index.
In fact, this strategy appears to be so obvious that it is often not motivated---let alone questioned---at all.

\subsection{Minimizing the BCE Learns the Posterior Marginals}\label{subsec:min_bce_post_marginals}

An ``ML style'' justification is to note that the \emph{expected} BCE between the bit vector $\bmb$ and its estimate $\bmf(\bmy)=(f_1,\dots,f_K)$ 
can be written as $\sum_k H(b_k|\bmy)+\Ex{\bmy}{D(p_{b_k|\bmy}\|f_k)}$, where $H(\cdot|\cdot)$ and $D(\cdot\|\cdot)$ 
are the conditonal and relative entropy. The expected BCE is thus minimized when the estimates $f_k(\bmy)$
equal the true posterior marginals $p_{b_k|\bmy}$.\footnote{This assumes that the transmitter is not trainable, 
so that $H(\bmb|\bmy)$ is a constant. See \cite{stark2019joint} for a discussion that includes trainable transmitters.}
Once the posterior is learned, simple thresholding (at $\frac12$) results in BER-optimal data detection. 
The \emph{expected} BCE is not available, but resorting to an empirical proxy through stochastic gradient descent 
is so common by now that it is often not even mentioned anymore. 

We now argue explicitly---using the framework of empirical risk minimization (ERM)---that minimizing the \emph{empirical} 
(as opposed to the expected) BCE can learn the true posterior marginals.
We do not claim that this result is ``novel,'' 
but an explicit derivation seems unavailable in the literature. 
In the ERM framework, one learns a function
\begin{align}
\textstyle
  \hat \bmf = \argmin_{\bmf\in\mathcal{F}} L(\bmf,\mathcal{D}), \label{eq:erm}
\end{align}
where
$\mathcal{F} \subseteq \{\bmf:\setY\to [0,\!1]^K\}$ is the set of admissible functions $\bmf=(f_1,\dots,f_K)$ and 
\begin{align}
\textstyle
  L(\bmf,\mathcal{D})= \sum_{n=1,\dots,N}  l_{\text{BCE}}(\bmb^{(n)},\bmf(\bmy^{(n)})),\label{eq:bce}
\end{align}
is the empirical risk, which here is induced by the BCE loss
\begin{align}
\textstyle
   \!\!l_{\text{BCE}}(\bmb,\bmf) = - \sum_{k=1}^K b_k \log(f_k) + (1\!-\!b_k) \log(1-f_k).\!\!\!
   \label{eq:lbce}
\end{align}
In principle, the empirical risk would be minimal if
\begin{align}
  \bmf(\bmy^{(n)})=\bmb^{(n)}, \quad n=1,\ldots,N.
  \label{eq:interpolate}
\end{align}
The optimal $\bmf$ would therefore make hard decisions on the training data set that---with hindsight---are 
always right. However, there are \textit{a priori} no restrictions on how such a function~$\bmf$ responds to an 
input $\bmy$ that is not contained in $\setD$: 
We are at the danger of overfitting.
ERM with a BCE loss may therefore be a reasonable strategy primarily in one of the following two settings: 
Either $\setF$ is ``inflexible'' or the range $\setY\ni\bmy$ is ``small'' compared to $\setD$.
In either case, \eqref{eq:interpolate} cannot be satisfied and overfitting is prevented.\footnote{It has 
been argued that learned systems may also generalize to new inputs \textit{even when} they achieve perfect
accuracy on the training dataset \cite{wyner2017explaining, belkin2018overfitting}.
An investigation of such settings is, however, beyond the scope of this paper.}
The first case is more relevant in practice but more difficult to analyze. 
We therefore focus on the second case, which we formalize through the following~assumption:
\begin{ass}
We assume that $\setD$ is large and representative of the underlying
posterior marginals $p_{b_k|\bmy}$ in the sense that, for some $0<\varepsilon<1$ and for all $k$ and all $(b,\bmy)\in \{0,1\}\times\setY$,
\begin{align}
\textstyle	\left|  p_{b_k|\bmy}(b=1|\bmy) -
	\frac{1}{|\setN(\bmy)|} \sum_{n\in\setN(\bmy)} b_k^{(n)}
    \right| \leq \varepsilon, 
\end{align}
where $\setN(\bmy)=\{ n\in\{1,\dots,N\} : \bmy^{(n)} = \bmy \}$.
\end{ass}

\begin{prop} \label{prop:bce}
Under Ass.~1, ERM with $\setF\!=\!\{\bmf\!:\!\setY\!\to\![0,1]^K\}$ and BCE loss
learns the posterior marginals up to~precision $\varepsilon$,
\begin{align}
	| p_{b_k|\bmy}(b=1|\bmy) - \hat f_k(\bmy)  | \leq \varepsilon, 
	 ~\forall \bmy\in\setY, ~k=1,\dots,K.
\end{align}
\end{prop}

\begin{showicassp}
The proof of this proposition (as well as of all following propositions) is included in the arXiv 
version \cite[Sec.\,7.1]{wiesmayr2022arxiv}.
\end{showicassp}
\begin{showarxiv}
The proof of this proposition (as well as of all following propositions) is shown in \fref{sec:proofs}.
\end{showarxiv}

It should be interesting to translate this result to the case where $\mathcal{Y}$ is uncountable 
but $\mathcal{F}$ is ``inflexible,'' or even to the interpolating case described in \cite{wyner2017explaining}.
We also note that, while the BCE is the most natural and probably most widely used loss 
in this context, it is by no means the only option. In fact, an analogous version of Prop.~\ref{prop:bce} holds 
for the mean square error (MSE) loss $l_\text{MSE}:\{0,1\}^K\times [0,1]^K, (\bmb,\bmf) \mapsto \|\bmb-\bmf\|_2^2/K$.

\begin{prop}\label{prop:mse}
Under Ass.~1, ERM with $\setF\!=\!\{\bmf\!:\!\setY\!\to\![0,1]^K\}$ and MSE loss
learns the posterior marginals up to~precision $\varepsilon$,
\begin{align}
	| p_{b_k|\bmy}(b=1|\bmy) - \hat f_k(\bmy) | \leq \varepsilon, ~\forall \bmy\in\setY, 
	~k=1,\dots,K.
\end{align}
\end{prop}

\subsection{Posterior vs. Posterior Marginals}\label{subsec:posterior_vs_marginals}
We now draw attention to a subtle but conceptually important point:
The loss in \eqref{eq:lbce} considers the sum of empirical BCEs between 
the individual components of $\bmb$ and $\bmf$, and we have shown that
this loss can be used to learn the posterior marginals $p_{b_k|\bmy}, k=1,\dots,K$. 
But this is not equivalent to learning the joint posterior $p_{\bmb|\bmy}$, 
since we do not learn the conditional dependencies between the different bits $b_k$. 
As a consequence of the summation of the component BCEs,
$\bmf$ approximates the posterior as a product of \emph{independent}~distributions. 
\begin{showarxiv}%
For an information-theoretic perspective, see also \fref{sec:mi}.
\end{showarxiv}
\begin{showicassp}%
For an information-theoretic perspective, see also \cite[Sec.\,7.2]{wiesmayr2022arxiv}.
\end{showicassp}

\section{Training for Block Error Rate}

\subsection{The Difference Between BER and BLER Optimality}
Learning to minimize the BLER in (block-)coded systems is \textit{not} tantamount with
learning to minimize the BER in those systems.
To see this, consider a (block-)coded system in which the bits $\bmb=(b_1,\dots,b_K)$ are encoded into codewords 
$\bmc=\text{enc}(\bmb)\in\setC$
for reliable data transmission. (In contrast to Sec.\,\ref{sec:ber_training}, we now
look at multiple bits from the \emph{same} data stream.)
Optimal (coded) BER is obtained when we decode on the basis of the posterior 
probabilities $p(b_k|\bmy)$, which---as we have seen---can be learned, e.g., with a BCE loss function:
\begin{align}
\textstyle
	\hat b_k = \argmax_{b_k\in\{0,1\}} p_{b_k|\bmy}(b_k|\bmy), ~~ k=1,\dots,K. \label{eq:ber_rule}
\end{align}
Perhaps surprisingly, this need not coincide with BLER-optimal decoding,
which is achieved by the decoding rule
\begin{align}
\textstyle
	\hat \bmb = \text{dec}(\argmax_{\bmc\in\setC} ~p_{\bmc|\bmy}(\bmc|\bmy)), \label{eq:bler_rule}
\end{align}
where $\text{dec}\!=\!\text{enc}^{-1}$ is the inverse mapping of the encoder. 
The reason is as follows: 
Even though the data bits $\bmb$ may be independent \textit{a priori}, 
their conditional distribution given the channel output, $p_{\bmb|\bmy}(\bmb|\bmy)$, 
is in general no longer so, 
$p_{\bmb|\bmy}(\bmb|\bmy)\neq \prod_{k=1,\dots,K} p_{b_k|\bmy}(b_k|\bmy)$.
We have the \mbox{following~result:}

\begin{prop}\label{prop:ber_opt_not_bler_opt}
Bit error rate (BER) optimal decoding in (block-) coded communication systems need not coincide with
block error rate (BLER) optimal decoding. 	
\end{prop}

Since the BCE and MSE loss learn the posterior~marginals instead of the joint posterior, 
they are inherently~aimed at solving the BER-optimal 
decoding problem \eqref{eq:ber_rule}, but not the BLER-optimal 
problem \eqref{eq:bler_rule} which is relevant \mbox{in practice.}

\subsection{Loss Functions for Block Error Rate Optimization}
\label{sec:bler_metrics}
We now propose several loss functions that aim at minimizing the BLER directly 
by penalizing the estimated bits $\bmb$ of a block in joint fashion 
instead of individually (by summation).

We use logits (often referred to as LLRs)
$\ell_k=\ell_k(\bmy)$ to represent the confidence that the $k$th bit $b_k$ is one or zero, 
respectively, where $\ell_k=+\infty$ means complete certainty that $b_k=1$, $\ell_k=-\infty$ means 
complete certainty that $b_k=0$, and $\ell_k=0$ means complete uncertainty. 
A straightforwardly obtained cost function that promotes joint (instead of individual) correctness of the decoded bits of a block is the \emph{product loss}
\begin{align}
\textstyle
    l_{\Pi}(\bmb,\bell) 
    = 1 - \!\prod_{k=1}^K\! \sigma((2b_k-1)\ell_k), 
\end{align}
where $\bmb=(b_1,\dots,b_K)$ are the labels,
$\bell(\bmy)=(\ell_1,\dots,\ell_K)$ the predictions, 
and $\sigma(x)=1/(1+e^{-x})$ is the logistic sigmoid.
The product loss is differentiable and satisfies
$0\leq l_{\Pi} \leq 1$, where $0$ is approached iff $(2b_k-1)\ell_k\to \infty$ 
for all $k$, and where $1$ is approached iff $(2b_k-1)\ell_k\to -\infty$ for at least one~$k$.
A practical issue with the product loss is that, since it is bounded, it can lead to the problem of vanishing gradients. 

A loss that promotes joint instead of individual correctness of the bits in a block 
while not being bounded is the \emph{Max loss} 
\begin{equation}
    l_{\max}(\bmb,\bell) = \max(x_1,\dots, x_K)
\end{equation}
with $x_k=l_{\mathrm{BCE}}(b_k,\rho(\ell_k))$, where $\rho(\cdot)$ maps from logits to probabilities. 
However, the max loss has the undesirable property that, for any given $(\bell,\bmb)$, only one of the partial derivatives 
with respect to $\ell_k$ is nonzero. 
We therefore also propose the usage of well-known smooth approximations to the Max loss. 
Among these are the \emph{SmoothMax loss} with parameter $\alpha$
(which we set to $\frac{1}{2}$ in our experiments)
\begin{equation}\textstyle
    \!\!\!l_{\mathrm{SM}}(\bmb,\bell; \alpha) =\! \sum_{k=1}^K x_k \exp(\alpha x_k)\big/\!\big({\sum_{k=1}^K \exp(\alpha x_k)}\!\big),\!\!
\end{equation}
the  \emph{LogSumExp loss} (normalized with $\gamma=K-1$) 
\begin{equation}
\textstyle
    l_{\mathrm{LSE}}(\bmb,\bell; \gamma) =  \log\left(\sum_{k=1}^K\exp(x_k)- \gamma \right),
\end{equation}
and the \emph{$p$-norm loss} for $p\geq1$ (with regularizer $\gamma=10^{-8}\!>\!0$) 
\begin{equation}
    l_{p}(\bmb,\bell; \gamma) = \left(x_1^p + \dots + x_K^p + \gamma \right)^{\frac1p}.
\end{equation}

\begin{remark}
A popular loss for learning end-to-end communication systems is the categorical CE (CCE) between the 
transmitted and guessed message \cite{dorner2017deep, aoudia2019model, song2022benchmarking}. By identifying
messages with the blocks of a block code, the CCE can be seen as a loss that optimizes the BLER. 
CCE-based learning, however, seems to be feasible only for very short blocks of $K\lessapprox8$~bits. 
\end{remark}

\section{SNR deweighted Training}
\label{sec:typestyle}
ML-assisted communication systems often learn a single set of parameters while operating
over a range of SNRs. To perform well over an entire range, training data should be sampled 
from the targeted SNR range. However, the aggregate loss of the training set will then be dominated 
by low-SNR data samples. Consequently, training will focus on low-SNR performance, because 
a small relative improvement at low SNR will affect the cost much more than a large relative improvement at high SNR. 
\fref{fig:normalized_loss_functions_vs_snr} showcases the issue by visualizing the average loss
when using an LDPC code with a classical BP decoder over an AWGN channel
for the different loss functions as a function of SNR (normalized such that the average loss at $0$\,dB is $1$). 
Evidently, the loss depends strongly on the SNR. In fact, the average losses closely mirror the bit/block error~rates.\footnote{
Note that the BER curve is shaped more like the BER losses (BCE/MSE), whereas the BLER curve is shaped 
more like the BLER losses.
This supports the insight that BCE or MSE do not optimally target the BLER.}
\begin{figure}
\begin{minipage}[b]{.49\linewidth}
  \centering
  \centerline{\includegraphics[width=4.1cm]{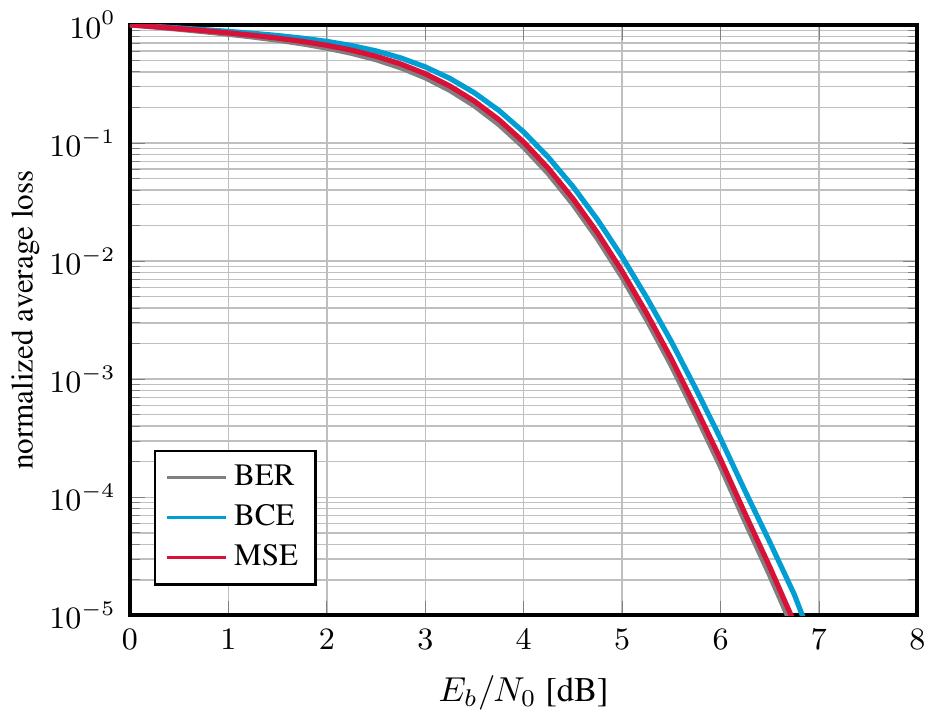}}
 \vspace{-0.2cm}
\end{minipage}
\begin{minipage}[b]{.49\linewidth}
  \centering
  \centerline{\includegraphics[width=4.1cm]{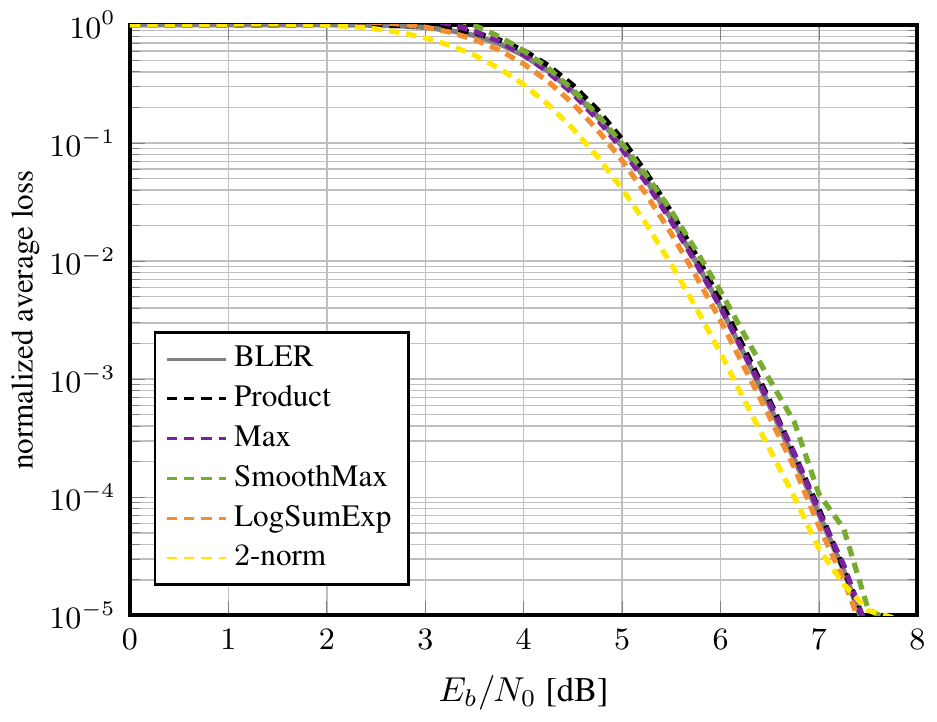}}
 \vspace{-0.2cm}
\end{minipage}
\hfill
\caption{The average loss of different BER (left) and BLER (right) losses is just as SNR-dependent as BER and BLER.}
\label{fig:normalized_loss_functions_vs_snr}
\end{figure}

To compensate for this effect, we propose \emph{SNR deweighted training}:
Training consists of multiple epochs with $M$ batches per epoch and $N$ Monte-Carlo (MC) samples 
per batch.
We partition every batch $\{1,\dots,N\}$ into $J\ll N$ sets $\setN(j),j=1,\dots,J$, each of which we associate with an SNR value that is selected from a uniform (in dB) grid which covers the desired range of operation.
The loss of the $m$th batch is defined as
\begin{align}
\textstyle
    L_m = \frac{1}{N}\sum_{j=1}^J\sum_{n\in\setN(j)} w^{(j)}l_m^{(n)}, \quad m=1,\dots,M,
\end{align}
where $l^{(n)}_m=l(\bmb^{(n)}_m,\bell^{(n)}_m)$ is the loss of the $n$th MC \mbox{sample} in the $m$th batch. 
The weights $w^{(j)}$ are initialized to $1$ and updated after every epoch: 
To balance the loss over the SNR range, we accumulate the loss over all samples with the same SNR, 
$l^{(j)}_\text{cum} = \sum_{n\in\setN(j)}\ \sum_{m=1}^M l_m^{(n)}.$
The weights for the next epoch are set to the inverse
cumulative losses, plus a constant $\delta\!>\!0$ that bounds the weight for stability:
$\tilde{w}^{(j)} = (l^{(j)}_\text{cum} + \delta)^{-1}.$
To avoid global loss scaling, we normalize the weights by dividing by the weight at the 
grid center, $\tilde{w}^{(\floor{J/2})}$, i.e.,
$w^{(j)} = \tilde{w}^{(j)}/\tilde{w}^{(\floor{J/2})}$,
before continuing~training.

Alternatively, one might also perform SNR deweighting by using the loss of a fixed baseline (e.g., a classical 
communication system) to deweight the training samples, instead of using the adaptive reweighting strategy 
described here.

\section{Simulation Results}
\label{sec:simulation_results}
We evaluate the utility of the different losses and of SNR~de-weighting 
through simulations in NVIDIA Sionna v0.12.1~\cite{Hoydis2022}.\begin{showicassp}\footnote{A second experiment on a simple trainable LDPC decoder for a single-input single-output (SISO) complex AWGN channel 
is included in the supplementary material of the extended arXiv version \cite{wiesmayr2022arxiv}.}
\end{showicassp}
\begin{showarxiv}\footnote{A second experiment on a simple trainable LDPC decoder for a single-input single-output 
(SISO) complex AWGN channel 
is included in Sec. \ref{ssec:siso_awgn}.}
\end{showarxiv}
We consider
a novel deep unfolded interleaved detection and decoding (DUIDD) receiver \cite{Wiesmayr2022} for a 5G MIMO-OFDM wireless system with 4 single-antenna UEs and one 16-antenna base station.
We use a short rate-matched $(80,60)$ 5G LDPC code based on a $(520,100)$ code with lifted base graph (BG)~$2$.
The coded bit stream is mapped to QPSK symbols, which are transmitted over a 3GPP UMa line-of-sight wireless channel. The channel is estimated using pilots and a least-square estimator with linear interpolation across frequency and time. 

In the first experiment, we consider the difference between BER and BLER performance with different losses
for training (\fref{fig:duidd_mimo_snr_range}), as well as with an untrained
``classical'' receiver \cite{Wiesmayr2022}.
We learn a single parameter set by training over a $[-10,10]$\,dB interval (without SNR deweighting). 
We start by pre-training the receiver for $2500$ batches of $N\!=\!200$ MC samples with~the BCE (or MSE) loss.
We then fine-tune the receiver by training with the respective loss functions for another $2500$ batches. 
Because we do not use SNR deweighting, the low-SNR region dominates training. 
The results show that the BER losses (BCE and MSE, solid) have the best BER-performance in the dominant low-SNR region (\circled{1}), 
but that the BLER losses (dashed) have superior BLER-performance at low-SNR (\circled{3}). 
Somewhat surprisingly, we observe that in the high-SNR regime---which is neglected during training, since 
we do not use SNR deweighting---the BLER losses outperform the BER losses in terms of BER (\circled{2}) as well as BLER (\circled{4}). 
The improvement in BLER-performance of the best BLER loss compared to the best BER loss 
is $0.62$\,dB at a BLER of $1$\% (\circled{4}).

In a second experiment, we select the product loss to consider the impact of different SNR training methods (\fref{fig:duidd_product}).
In the left figure, we compare na\"ive training over a large SNR range
of $[-10,10]$\,dB (R$_{[-10,10]\text{dB}}$) with SNR deweighted training over that same range
(DW$_{[-10,10]\text{dB}}$), as well as with training at a single SNR point at $-5$\,dB (P$_{-5\text{dB}}$)
.\footnote{For pre-training, we applied BCE loss and trained on the same SNR range or point as in the latter refinement step, respectively, but without deweighting.}
The results show that na\"ive training over the range, as well as training only at
a single low-SNR point achieves good relative performance at low SNR (\circled{5})
but comparably bad performance at high SNR (\circled{7}). In fact, training only 
at a low-SNR point leads to a complete breakdown at very high SNRs (in this experiment).
Finally, SNR deweighted training achieves well-balanced performance 
even when training over such a large SNR range. 
SNR deweighting outperforms na\"ive training over the range by 0.54\,dB at a BLER of 1\%  (\circled{7}).

In the right figure of \fref{fig:duidd_product}, we perform the same experiment, but training
only over a smaller range of SNRs ($[-4,6]$~dB) in the BLER waterfall region (or a single point therein). 
The results show significant convergence between the different training methods 
in this case.
Training on a single SNR point (P$_{5\text{dB}}$) achieves slightly better performance than its competitors at high SNR (\circled{9}), but performs worse at low SNR (\circled{8}).
SNR deweighted training still enjoys a (tiny) advantage over na\"ive training over the SNR range at high SNR (\circled{9}), while na\"ive training over the SNR range enjoys an (even tinier) advantage at the low-SNR end of the waterfall (\circled{8}).
These results highlight that SNR plays an important role in training:
Na\"ive training over a range focuses excessively on~low SNRs.
Training at a single SNR sometimes works well, but sometimes leads to bad surprises. 
In contrast, SNR deweighted training seems to be robust and provide uniformly good performance. 

The gains that are afforded in these experiments by BLER specific losses and by SNR deweighting 
($0.62$\,dB and $0.54$\,dB, respectively) may be modest. 
However, we emphasize that these gains are not caused by a more elaborate receiver or more training data, 
but simply by using a more appropriate loss function. As such, they are effectively available \emph{for free}.

\begin{figure}[t]\hspace{-3.5mm}
\begin{minipage}[b]{.54\linewidth}
  \centering
  \centerline{\includegraphics[width=4.1cm]{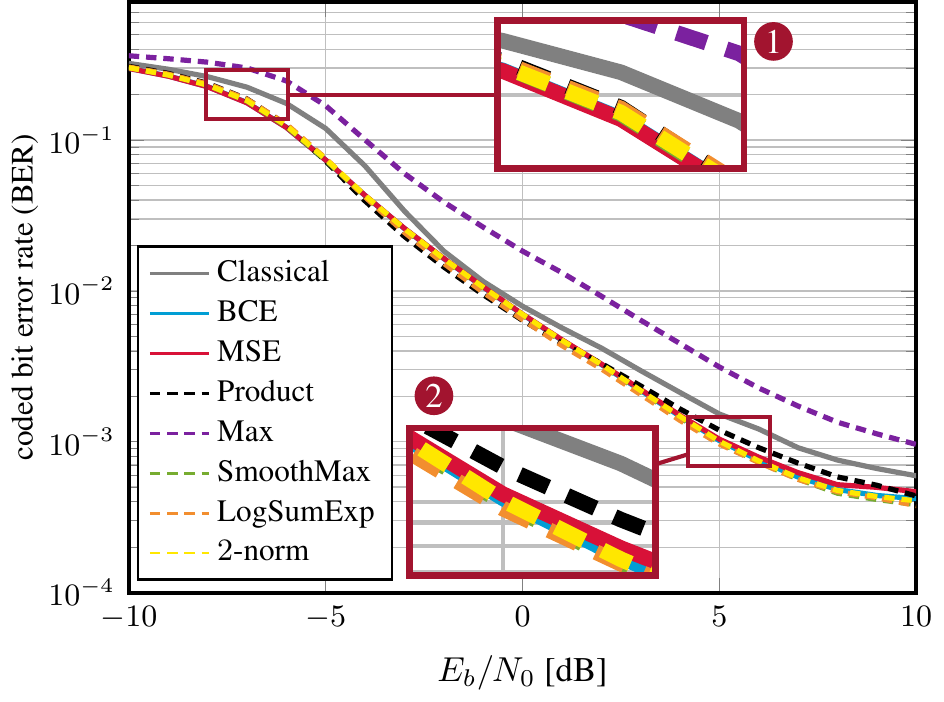}} 
 \vspace{-0.3cm}
\end{minipage}
\hspace{1mm}
\begin{minipage}[b]{.44\linewidth}
  \centering
\centerline{\includegraphics[width=4.1cm]{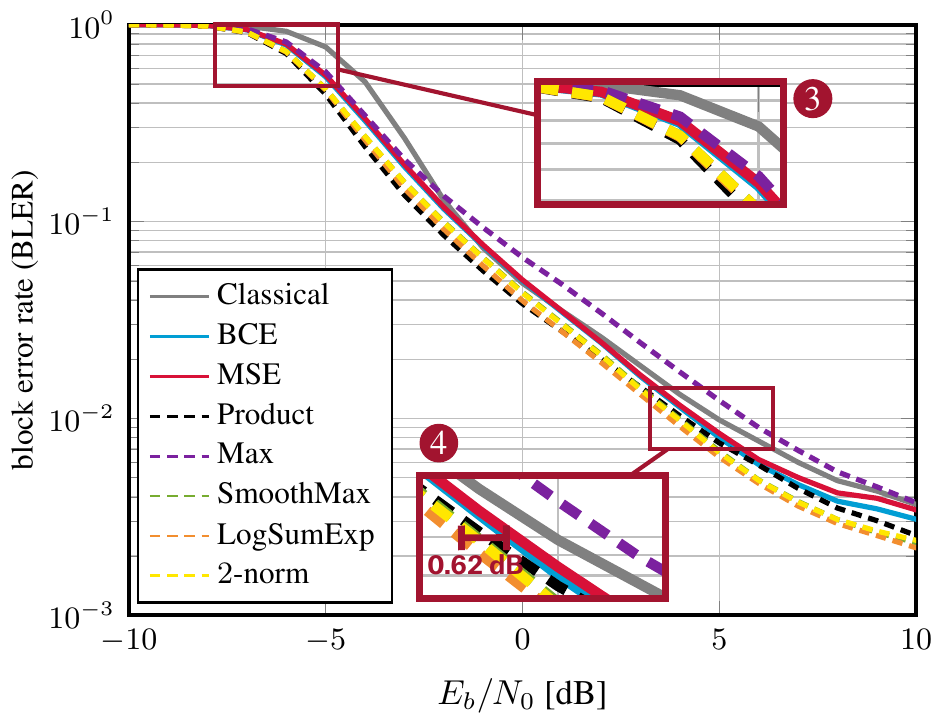}}
 \vspace{-0.3cm}
\end{minipage}
\hfill
\caption{Contrast between BER (left) and BLER (right) when training DUIDD with different losses.} %
\label{fig:duidd_mimo_snr_range}
\end{figure}
\begin{figure}[t]
\begin{minipage}[b]{.54\linewidth}
  \centering
  \centerline{\includegraphics[width=\linewidth]{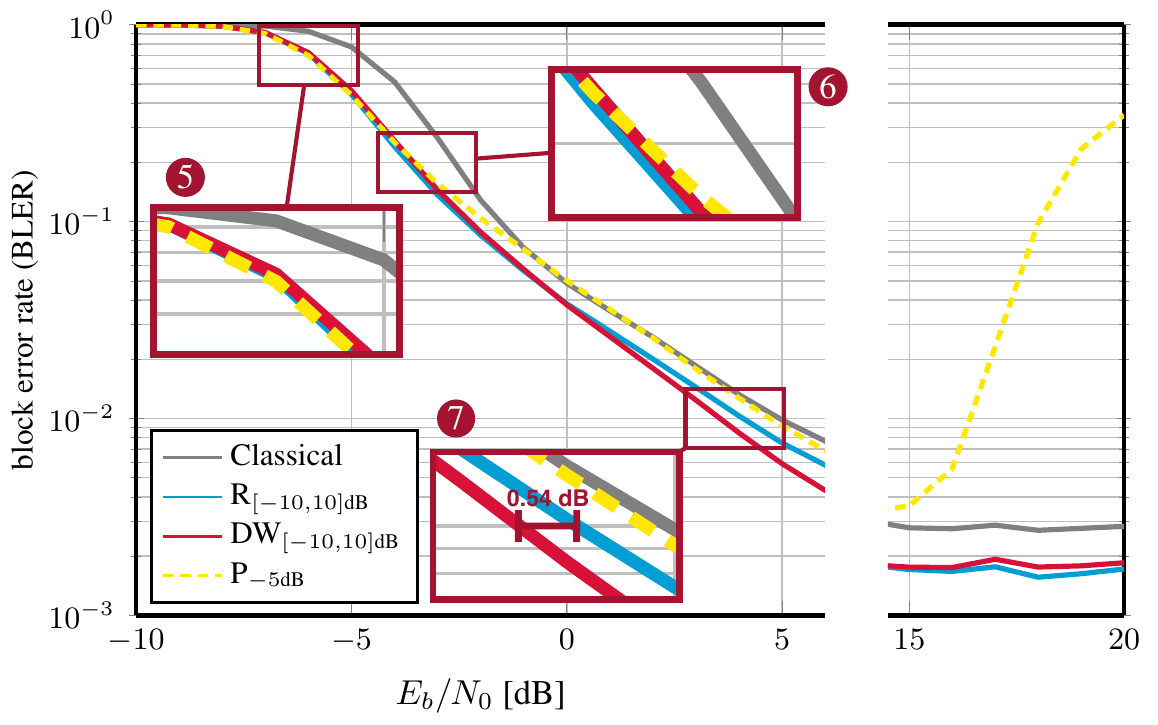}}
 \vspace{-0.3cm}
\end{minipage}
\begin{minipage}[b]{.44\linewidth}
  \centering
  \centerline{\includegraphics[width=\linewidth]{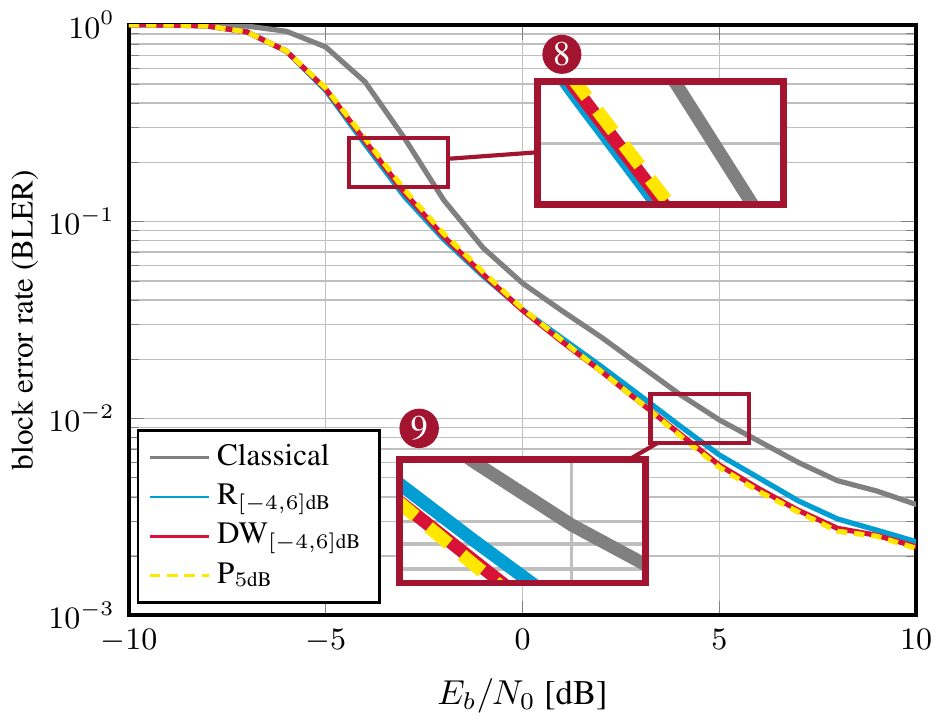}}
 \vspace{-0.3cm}
\end{minipage}
\hfill
\caption{Impact of training at a single SNR vs. over a range of SNRs vs. SNR deweighting, for DUIDD with Product loss.}
\label{fig:duidd_product}
\vspace{-0.3cm}
\end{figure}

\vspace{-0.125cm}
\section{Conclusions}\vspace{-0.125cm}
We have turned the spotlight on the impact that different loss functions and SNRs have on the training of ML-assisted communication
systems. Seemingly obvious losses, such as empirical BCE, turn out to be suboptimal for minimizing the BLER 
and are outperformed by BLER-specific losses. 
We have also shown that na\"ive training over a range of SNRs will focus excessively on the low-SNR (high-loss) region 
and neglect high-SNR performance. To compensate for this effect, we have proposed SNR deweighting.
The findings of this paper are not meant as final answers to the question of \emph{how} to train communication systems, 
but rather as a starting point to some of the relevant issues and considerations. 

\begin{showarxiv}
\section{Supplementary Material} \label{sec:supplement}

The contents of this section supplement a paper that will be presented at the 2023 IEEE International Conference on Acoustics, Speech, and Signal Processing (ICASSP). 

\subsection{Proofs} \label{sec:proofs}
\subsubsection{Proof of \fref{prop:bce}}
We start by rewriting \eqref{eq:bce} as
\begin{align}
  &-  \sum_{k=1}^K \sum_{\bmy\in\mathcal{Y}} \sum_{n\in\setN(\bmy)}
  \Big[ b_k^{(n)} \log(f_k(\bmy)) \nonumber \\
  & \qquad \qquad \qquad \qquad + (1-b_k^{(n)}) \log(1-f_k(\bmy)) \Big]  \\
  &= -  \sum_{k=1}^K \sum_{\bmy\in\mathcal{Y}} \bigg[ \log(f_k(\bmy)) \sum_{n\in\setN(\bmy)} 
    b_k^{(n)} \nonumber\\
    & \qquad \qquad \qquad +
   \log(1-f_k(\bmy)) \sum_{n\in\setN(\bmy)} (1-b_k^{(n)}) \bigg]
   \nonumber \\
  &= - \sum_{k=1}^K \sum_{\bmy\in\mathcal{Y}} |\setN(\bmy)| \Big[ \tilde p_k(\bmy) \log(f_k(\bmy)) \nonumber \\
  & \qquad \qquad \qquad \qquad \quad + (1-\tilde p_k(\bmy))\log(1-f_k(\bmy))  \Big], \label{eq:result}
\end{align}
where in \eqref{eq:result} we defined $\tilde p_k(\bmy) = \frac{1}{|\setN(\bmy)|}\! \sum_{n \in \setN(\bmy)} \!b_k^{(n)}\!$,
which satisfies $|\tilde p_k(\bmy) - p_{b_k|\bmy}(b\!=\!1|\bmy)| \leq \varepsilon$.
To analyze the minimization of \eqref{eq:result}, we consider the terms in $k$ and $\bmy$ individually.
The function 
\begin{align}
  h_q: [0,1]\to \reals,~ x \mapsto -q \log(x) - (1-q) \log(1-x)
\end{align}
parametrized by $q\in[0,1]$ is minimized for $x=q$. By term-wise minimization of \eqref{eq:result}, we
see that the $\hat \bmf$ which minimizes \eqref{eq:result} (and thus also \eqref{eq:erm}) satisfies
$\hat f_k(\bmy) = \tilde p_k(\bmy)$ for all $\bmy\in\setY$
and thus gives the true posterior marginals up to precision $\varepsilon$.~\hfill $\blacksquare$

\subsubsection{Proof of \fref{prop:mse}}
The proof is very similar to the one of \fref{prop:bce}. We start by rewriting the aggregate MSE loss as
\begin{align}
    &- \frac1K \sum_{k=1}^K \sum_{\bmy\in\mathcal{Y}} \sum_{n\in\setN(\bmy)}
    \left( b_k^{(n)} - f_k(\bmy) \right)^2
\end{align}
and consider the minimization of the terms in $k$ and $\bmy$ individually (while dropping the $\frac1K$ prefactor). 
These terms  $\sum_{n\in\setN(\bmy)}( b_k^{(n)} - f_k(\bmy) )^2$ are convex, so that the minimizing
$f_k(\bmy)$ can be found by finding the zero of the derivative, 
\begin{align}
    \frac{\partial}{\!\partial f_k(\bmy)} \sum_{n\in\setN(\bmy)} \!\! \left( b_k^{(n)} - f_k(\bmy) \right)^2
    &= 2 \sum_{n\in\setN(\bmy)} \!\! \left( b_k^{(n)} - f_k(\bmy) \right) \nonumber \\
    &= 0, 
\end{align}
which gives 
\begin{align}
    f_k(\bmy) &= \frac{1}{|\setN(\bmy)|} \sum_{n \in \setN(\bmy)} b_k^{(n)} \\
    &= p_{b_k|\bmy}(b=1|\bmy) + e_r
\end{align}
for some residual $e_r$ with $|e_r|\leq \varepsilon$.~\hfill~$\blacksquare$
\vspace{-2mm}

\subsubsection{Proof of \fref{prop:ber_opt_not_bler_opt}}
We prove the result with an example: Blocks of two  independent, equiprobable bits $b_1,b_2$ 
are encoded into four codewords: $\bmc_1=\text{enc}((b_1,b_2)=(0,0))$,
\mbox{$\bmc_2 = \text{enc}((b_1,b_2)=(0,1))$,} $\bmc_3 = \text{enc}((b_1,b_2)=(1,0))$, 
and $\bmc_4 = \text{enc}((b_1,b_2)=(1,1))$.
Assume that the channel output $\bmy$ induces the following posterior for the transmitted codeword: 
$p(\bmc_1|\bmy)\!=\!0.2, p(\bmc_2|\bmy)\!=\!0.35, p(\bmc_3|\bmy)\!=\!0.4$, and $p(\bmc_4|\bmy)\!=\!0.05$. 
It directly follows that block error rate optimal decoding according to \eqref{eq:bler_rule}
yields $(\hat b_1, \hat b_2)=(1,0)$, with the posterior probability of this being the transmitted block equal to $0.4$. 
However, the posteriors of the individual bits are
\begin{align}
	p_{b_1|\bmy}(b_1=1|\bmy)&=p_{\bmc|\bmy}(\bmc_3|\bmy) + p_{\bmc|\bmy}(\bmc_4|\bmy) = 0.45 \\
	p_{b_2|\bmy}(b_2=1|\bmy)&=p_{\bmc|\bmy}(\bmc_2|\bmy) + p_{\bmc|\bmy}(\bmc_4|\bmy) = 0.25.	
\end{align}
Bit error rate optimal decoding on the basis of \eqref{eq:ber_rule} thus gives $(\hat b_1, \hat b_2)=(0,0)$, 
with a block error probability of 0.8.
\hfill $\blacksquare$

\subsection{A Word About Mutual Information} 
\label{sec:mi}

In \fref{sec:ber_training}, we have motivated the use of the BCE or MSE loss for ML-assisted communications 
by arguing that it can learn the posterior marginals $p_{b_k|\bmy}$ and is therefore BER-optimal. 
Communication theorists, however, might be more immediately interested in learning functions  $f_k:\mathcal{Y}\to [0,1]$ 
that maximize the mutual information  $I(b_k;f_k(\bmy))$ 
(e.g., when eyeing to the combination with FEC). 
The issue with such an objective is that mutual information is invariant under bijections.  
That is, if $\hat{f}_k$ maximizes $I(b_k;f_k(\bmy))$, then so does $g\circ \hat{f}_k$ for any bijection $g:[0,1]\to[0,1]$. 
What matters in practice is to learn a function that is not only informative about $b_k$, but that also makes this information accessible in a predefined way.\footnote{For example, a soft-input LDPC decoder typically expects logits to perform message passing. Although the mapping $g(f_k(\bmy))=1-f_k(\bmy)$ would lead to the same MI, the decoder would not work with such $g$.}
Maximizing the MI achieves the former, but has no bearing on the latter. 
However, by learning the posterior marginals $p_{b_k|\bmy}$, one can achieve both goals at once. 
In particular, the posterior captures all the information about $b_k$ contained in $\bmy$.
Again, this result is hardly ``new,'' but we feel that it is worth spelling out explicitly.

\begin{prop}
	The binary posterior $\hat f_k: \bmy \mapsto p_{b_k|\bmy}(b=1|\bmy)$ maximizes the mutual information,
	\begin{align}
	\textstyle
		\hat f_k \in \argmax_{f_k:\setY \to [0,1]} I(b_k;f_k(\bmy)),
	\end{align}
and satisfies $I(b_k;\hat f_k(\bmy)) = I(b_k;\bmy)$.
\end{prop}
\begin{proof} 
\sloppy
By the data-processing inequality, $I(b_k;\bmy) \geq I(b_k;f(\bmy))$ for any $f$. 
Equality holds iff $f(\bmy)$ is a sufficient statistic relative to $p_{b_k|\bmy}(b_k|\bmy)$ \cite{Cover2006}.
By the Fisher–Neyman factorization theorem \cite{Silvey1975}, $t(\bmy)\triangleq p_{b_k|\bmy}(b=1|\bmy)$ 
is sufficient if there exist functions $q(b|\cdot), r(\cdot)$ such that $p_{b_k|\bmy}$ can be factorized as
\begin{align}
  p_{b_k|\bmy}(b|\bmy) = q\big(b\,|\,t(\bmy)\big)\, r(\bmy), ~\forall (b,\bmy) \in \{0,1\}\times\setY.
\end{align}
This criterion is satisfied for $r(\bmy)\equiv1$, and for $q(b|t)=t$ if $b=1$ and $q(b|t)=1-t$ if $b=0$. 
\fussy
\end{proof}

As a result, the posterior marginals $(p_{b_1|\bmy},\dots,p_{b_K|\bmy})$ also maximize 
the bitwise mutual information (BMI) \cite{Cammerer2020}
\begin{align}
\textstyle
    \sum_{k=1}^K I(b_k; f_k(\bmy))
\end{align}
among all functions $\bmf=(f_1,\dots,f_K)$, since they satisfy the inequality
\begin{align}
\textstyle
    \sum_{k=1}^K I(b_k; f_k(\bmy)) \leq \sum_{k=1}^K I(b_k; \bmy) 
\end{align}
with equality. The BMI was proven to be an achievable rate for bit-metric decoding in \cite{bocherer2017achievable}.
However, this does not imply that the posterior marginals maximize the mutual information between $\bmb$ and~$\bmf(\bmy)$. 
In particular, the posterior marginals $(p_{b_1|\bmy},\dots,p_{b_K|\bmy})$ do in general \emph{not} satisfy the inequality 
\begin{align}
    I(\bmb; \bmf(\bmy)) \leq I(\bmb; \bmy) 
\end{align}
with equality, since they do not capture the conditional dependencies between the individual bits.
These considerations give us an information theoretic perspective on why learning 
the posterior marginals can be ``optimal'' in terms of getting the individual bits right, 
but is not necessarily optimal in terms of getting the entire bit vector right.

\subsection{Simulation Results: SISO AWGN Channel}
\label{ssec:siso_awgn}

To show that the effects observed in \fref{sec:simulation_results} were not simply due to that one particular
communication scenario and/or receiver architecture, this section reports experiments for a second communication 
scenario. We consider a single-input single-output (SISO) circularly symmetric complex additive white Gaussian noise (AWGN) channel 
with a trainable LDPC receiver. This scenario is comparably much simpler than that of \fref{sec:simulation_results}: the channel model is much simpler, and the receiver has much less learnable ``structure.'' 
Both of these factors lead us to expect already in advance that there are much smaller gains
available for learned systems (and much smaller performance differences between different training 
methods) than for the scenario in \fref{sec:simulation_results}. 
What we are primarily interested, then, is whether the relations between the performances of
different training methods
are consistent with those of \fref{sec:simulation_results}, and not whether the observed effects 
are as pronounced.

\begin{figure}[t]
\begin{minipage}[b]{.49\linewidth}
  \centering
  \centerline{\includegraphics[width=4.1cm]{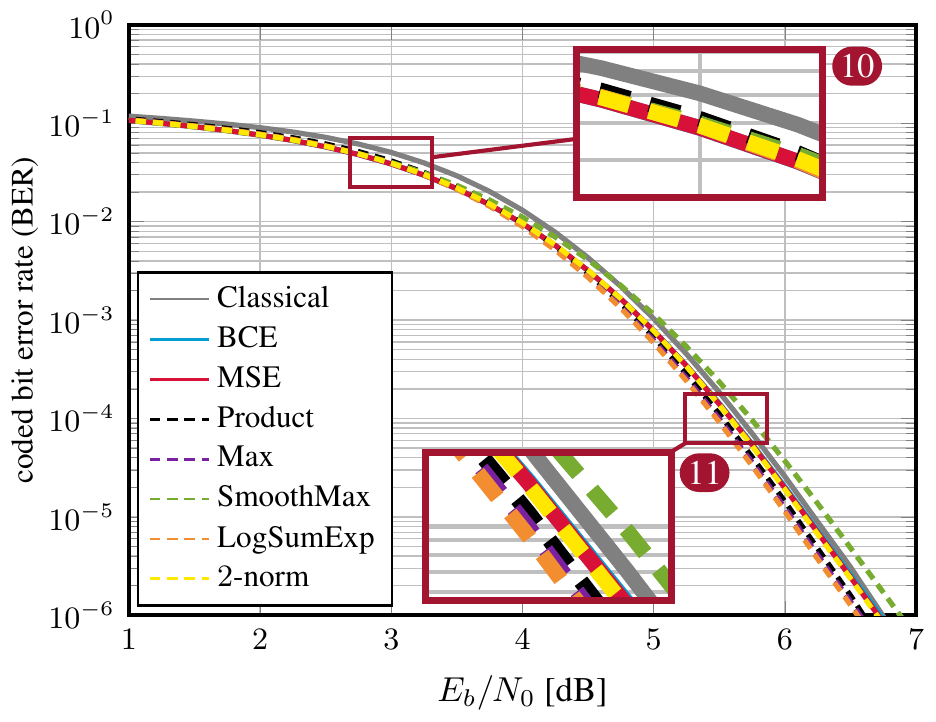}}
 \vspace{-0.25cm}
\end{minipage}
\begin{minipage}[b]{.49\linewidth}
  \centering
  \centerline{\includegraphics[width=4.1cm]{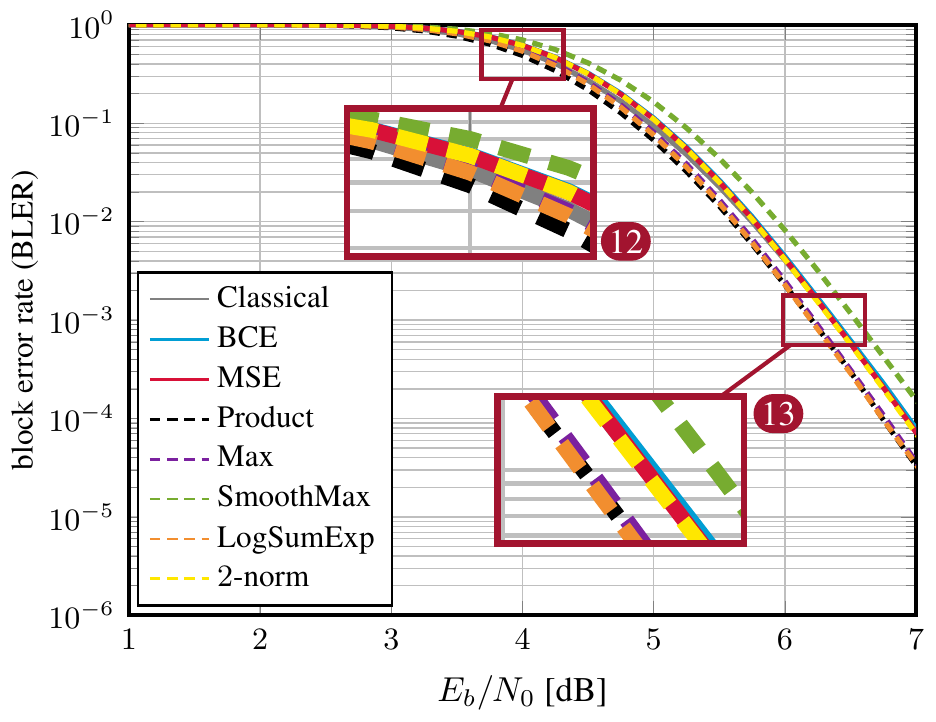}}
 \vspace{-0.25cm}
\end{minipage}
\hfill
\caption{Contrast between BER (left) and BLER (right) when training SISO AWGN receiver with different losses.}
\label{fig:siso_awgn_snr_range}
\end{figure}

\begin{figure}[t]
\begin{minipage}[b]{.49\linewidth}
  \centering
  \centerline{\includegraphics[width=4.1cm]{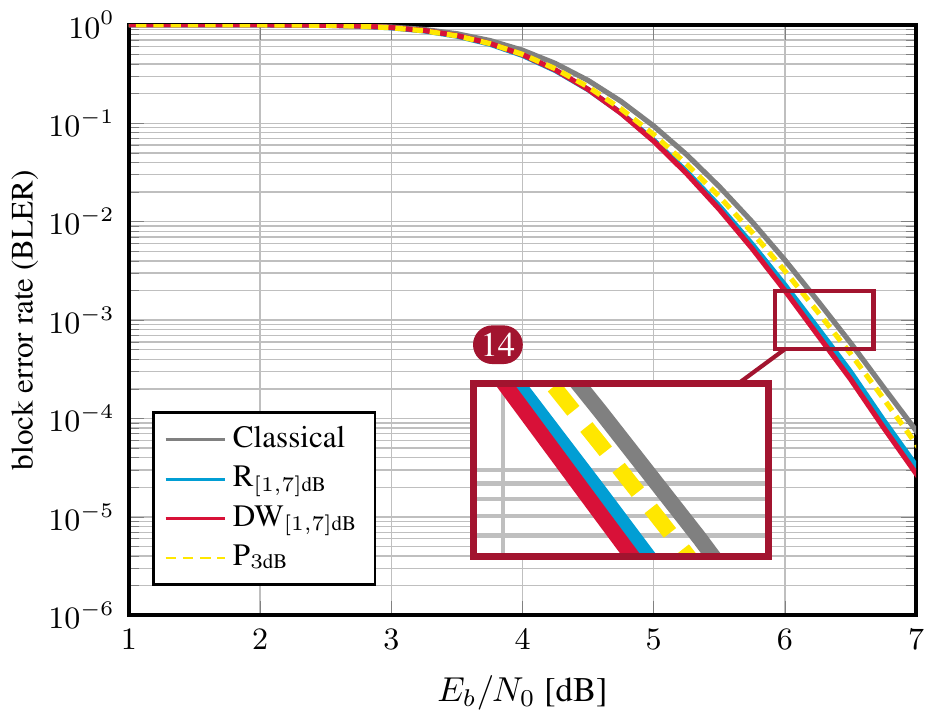}}
 \vspace{-0.25cm}
\end{minipage}
\begin{minipage}[b]{.49\linewidth}
  \centering
  \centerline{\includegraphics[width=4.1cm]{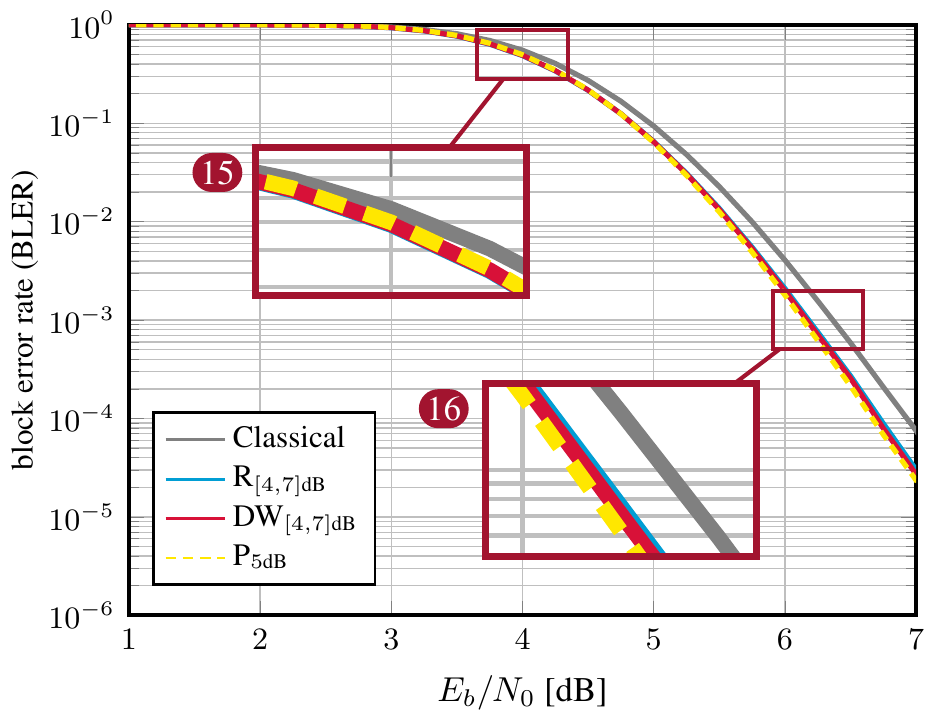}}
 \vspace{-0.25cm}
\end{minipage}
\hfill
\caption{Impact of training at a single SNR vs. over a range of SNRs vs. SNR deweighting, SISO AWGN with Product loss.}
\label{fig:siso_product_training}
\vspace{-0.1cm}
\end{figure}

We use a rate-matched $(360,300)$ 5G LDPC code based on~a $(952,308)$ code with lifted base graph (BG) $1$ \cite{Richardson2018}.
The code bits are interleaved as in \cite[Sec. 5.4.2.2]{ETSI5G} and mapped to QPSK symbols, which are sent over 
an AWGN channel. 
The decoder calculates logits
via max-log demapping and performs $5$ min-sum message passing (MP) iterations.
Using deep unfolding, we train LDPC belief propagation (BP) check-to-variable (C2V) node message damping \cite{Som2010} and the edge weights of the lifted BG.
Unlike \cite{nachmani2018deep, Lian2018}, we use individual damping values but one set of edge weights for all MP iterations.
We also dampen the check-node updates by subtracting the incoming V2C messages. 
We learn a single set of parameters by training over a $[1,7]$\,dB SNR range (no SNR deweighting), 
using the same pre-training/training procedure and the same number of training samples as in \fref{sec:simulation_results}.

In \fref{fig:siso_awgn_snr_range}, we consider the difference between BER and BLER performance
with different losses for training, as well with an untrained ``classical'' BP receiver. 
The results show that the BER losses (BCE and MSE) have the best BER-performance
(albeit with minuscule margins, \circled{10})
at low SNR which, without deweighting, dominates training.
Meanwhile, the BLER losses have better BLER-performance in the training-dominating low-SNR region 
(\circled{12}). Somewhat surprisingly---but consistent with the results for the DUIDD receiver
(\fref{sec:simulation_results})---the BLER losses outperform the BER losses in the high-SNR regime, 
which is neglected during training, in terms of BER (\circled{11}) as well as BLER (\circled{13}).
The difference in BLER-performance between a BLER loss and a BER loss 
can be as much as $0.18$\,dB at a BLER of $0.1$\%~(\circled{13}). While this margin is smaller than 
for DUIDD, 
the BLER loss functions are still able to noticeably improve the BLER performance of learned receiver
over an untrained (classical) receiver. In stark contrast, the BLER performance of the receivers
trained with the BCE or MSE loss is identical (\circled{13}) or \emph{worse} (\circled{12})
than that of a completely untrained receiver. 

Analogously to the experiment of \fref{fig:duidd_product}, we now 
consider the impact of different SNR training methods on the SISO AWGN receiver, focusing again on the 
example of the Product loss (\fref{fig:siso_product_training}).
In the left figure, we compare na\"ive training over a large SNR range of $[1,7]$\,dB (R$_{[1,7]\text{dB}}$) with SNR deweighted 
training over that same range (DW$_{[1,7]\text{dB}}$), as well as with training at a single SNR point at $3$\,dB (P$_{3\text{dB}}$).
Just as for the DUIDD receiver (\fref{sec:simulation_results}), the results show that na\"ive training over the range, 
as well as at a single low-SNR point achieves good relative performance at low SNR but comparably bad performance (or, in 
case of P$_{3\text{dB}}$, very bad performance) 
at high SNR (\circled{14}). Meanwhile, SNR deweighted training 
achieves good performance at high SNR as well as at low SNR.
In the right figure of \fref{fig:siso_product_training}, we again only train over a smaller range of SNRs ($[4,7]$\,dB)
in the BLER waterfall region (or a single point therein). Na\"ive training over the range again performs worst 
at high SNR (\circled{16}) and best at low SNR (\circled{15}). Training at a single SNR performs mildly better at high SNR while SNR deweighted training
performs mildly better at low SNR. 
In summary, these results largely confirm the results of \fref{sec:simulation_results}.

\end{showarxiv}

\balance

\end{document}